\documentclass[letterpaper, 10 pt, conference]{ieeeconf}
\IEEEoverridecommandlockouts
\usepackage[T1]{fontenc}

\usepackage{cite}
\usepackage{amsmath,amssymb,amsfonts}
\allowdisplaybreaks
\usepackage{bm}
\usepackage{algorithm}
\usepackage{algpseudocode}
\usepackage{graphicx}
\usepackage{textcomp}
\usepackage{xcolor}
\usepackage{tabularx}
\usepackage{graphicx}
\usepackage{array}

\usepackage{subcaption}
\usepackage{adjustbox}
\usepackage{mathtools}
\usepackage{balance}
\usepackage{hyperref}
\usepackage{authblk}
\usepackage{siunitx}

\usepackage{icpslab} 

\def\Bc{{\mathcal B}}

\def\Dc{{\mathcal D}}

\def\Hc{{\mathcal H}}

\def\Lc{{\mathcal L}}

\def\Nbb{{\mathbb N}}

\def\pbf{{\mathbf p}}

\def\Rc{{\mathcal R}}
\def\Rbb{{\mathbb R}}

\def\Sc{{\mathcal S}}

\def\ubf{{\mathbf u}}

\def\xbf{{\mathbf x}}

\def\ybf{{\mathbf y}}

\def\pa{{\partial}}

\def\Om{{\Omega}}

\def\0{{\bf 0}}

\newcommand{\bitem}{\begin{itemize}}
\newcommand{\eitem}{\end{itemize}}
\newcommand{\btabular}{\begin{tabular}}
\newcommand{\etabular}{\end{tabular}}
\newcommand{\bcenter}{\begin{center}}
\newcommand{\ecenter}{\end{center}}
\newcommand{\bea}{\begin{eqnarray}}
\newcommand{\eea}{\end{eqnarray}}
\newcommand{\bean}{\begin{eqnarray*}}
\newcommand{\eean}{\end{eqnarray*}}

\newcommand{\ba}{\left. \begin{array}}
\newcommand{\ea}{\\ \end{array} \right.}
\newcommand{\bab}{\left[ \begin{array}}
\newcommand{\eab}{\\ \end{array} \right]}
\newcommand{\bap}{\left( \begin{array}}
\newcommand{\eap}{\\ \end{array} \right)}
\newcommand{\bbm}{ \begin{bmatrix}}
\newcommand{\ebm}{\\ \end{bmatrix} }
\newcommand{\bear}{\begin{array}}
\newcommand{\eear}{\\ \end{array}}

\newcommand{\bs}{\boldsymbol}

\newcommand{\ra}{\rightarrow}

\font\myownfont=cmr17 scaled \magstep5
\def\psfancypar#1#2{\def\biginitial#1{{\myownfont#1}}%
  \def\makeinitial#1{\setbox8\hbox{\strut\vbox to 1.3ex
    {\hbox{\biginitial#1}\vskip -4pc plus 3.5pc minus 3.5pc}}}%
  \makeinitial#1%
  \ifdim\parindent>1.3\wd8\dimen8=\parindent
     \else\dimen8=1.3\wd8\fi
  \hangindent=\dimen8\hangafter=-2
  \noindent
  \strut\hskip-1\dimen8\box8{\sc#2}}%

\newcounter{subequation}
\def\beasub{\addtocounter{equation}{+1}
\setcounter{subequation}{\value{equation}}
\setcounter{equation}{0}
\renewcommand{\theequation}{\arabic{subequation}\alph{equation}}
\begin{eqnarray}}
\def\eeasub{\end{eqnarray}
\setcounter{equation}{\value{subequation}}
\renewcommand{\theequation}{\arabic{equation}}}

\newtheorem{theorem}{Theorem}
\newtheorem{lemma}[theorem]{Lemma}
\newtheorem{proposition}[theorem]{Proposition}

\newtheorem{definition}{Definition}
\newtheorem{remark}{Remark}
\newtheorem{assumption}{Assumption}

\newcommand{\xeq}{\xbf_\mathrm{eq}}

\newcommand{\gradH}{\nabla_{\!\xbf} H}
\newcommand{\FNN}{\mathbf{NN}_H}

\makeatletter
\newcommand{\linebreakand}{%
\end{@IEEEauthorhalign}
\hfill\mbox{}\par
\mbox{}\hfill\begin{@IEEEauthorhalign}
}
\makeatother

\title{Structure- and stability-preserving learning-based model order reduction\\for port-Hamiltonian systems}

\begin{document}
\title{Structure- and Stability-Preserving Learning of Port-Hamiltonian Systems}
\author{Binh Nguyen, Nam T. Nguyen and  Truong X. Nghiem
\thanks{This material is based upon work supported by the U.S.\ National Science Foundation under Grant No.~2514584.}
\thanks{All authors are with the Department of Electrical and Computer Engineering, College of Engineering and Computer Science, University of Central Florida, Orlando, FL 32816, USA.}
}
\maketitle

\begin{abstract}
This paper investigates the problem of data-driven modeling of port-Hamiltonian systems while preserving their intrinsic Hamiltonian structure and stability properties. 
We propose a novel neural-network-based port-Hamiltonian modeling technique that relaxes the convexity constraint commonly imposed by neural network-based Hamiltonian approximations, thereby improving the expressiveness and generalization capability of the model. 
By removing this restriction, the proposed approach enables the use of more general non-convex Hamiltonian representations to enhance modeling flexibility and accuracy.
Furthermore, the proposed method incorporates information about stable equilibria into the learning process, allowing the learned model to preserve the stability of multiple isolated equilibria rather than being restricted to a single equilibrium as in conventional methods. 
Two numerical experiments are conducted to validate the effectiveness of the proposed approach and demonstrate its ability to achieve more accurate structure- and stability-preserving learning of port-Hamiltonian systems compared with a baseline method.

\end{abstract}

\section{Introduction}
Port-Hamiltonian systems (PHS) provide a powerful framework for modeling, analysis, 
and control of complex physical systems by emphasizing the flow and storage of energy 
\cite{vanderschaftPortHamiltonianSystemsNetwork2004,    
rashadTwentyYearsDistributed2020}. 
Originating from classical Hamiltonian mechanics, PHS extends the Hamiltonian concept by incorporating external interactions through ports, allowing explicit representation of energy exchange with the environment.  This formulation inherently preserves fundamental physical properties, such as passivity and conservation laws, making it particularly suitable for multidisciplinary applications involving interconnected subsystems.

Applications of PHS span various fields, including mechanical systems, 
electrical circuits, fluid dynamics, robotics, and thermal systems
\cite{zhongPortHamiltonianControlFramework2022,
rashadEnergyAwareImpedance2022,
altawaitanHamiltonianDynamicsLearning2024
}. 
Their versatility enables engineers and scientists to systematically represent complex, 
interconnected structures, ensuring physically consistent modeling that simplifies system 
integration and facilitates robust controller design.
A port-Hamiltonian neural network (PHNN) is a neural network (NN) based machine learning framework inspired 
by PHS \cite{desaiPortHamiltonianNeuralNetworks2021, eidnesPseudoHamiltonianNeuralNetworks2023,
duongPortHamiltonianNeuralODE2024}.
By integrating NNs with port-Hamiltonian structures, PHNNs explicitly encode energy conservation and dissipation principles, enabling robust modeling and 
prediction of dynamical behaviors and enhancing stability, interpretability, and generalization capabilities.
In particular, an appropriate composition of multiple PHNNs can capture large-scale, complex dynamical behaviors while preserving energy conservation principles.
This makes them particularly valuable for applications involving mechanical systems \cite{o2025port}, electrical systems \cite{bartel2022port}, and control systems
\cite{nageshrao2015port}, where accurate modeling of energy transfer and system interactions is essential.

Since the Hamiltonian of a PHS is typically unknown and cannot be directly sampled, PHNNs provide a promising data-driven approach for learning PHS models solely from data consisting of state variables and system inputs \cite{rettberg2025data,cherifi2025nonlinear}.
However, to the best of our knowledge,
limited progress has been made in making PHNNs 
accurately preserve stable equilibria of the original system in the learned model.
Preserving stable equilibria is crucial as they characterize long-term behaviors of the system.
In many physical systems, stable equilibria correspond to desired operating points, such as rest configurations or energy-minimizing states.

Accurately capturing these equilibria ensures that the learned model reproduces correctly steady-state behaviors of the original system;
otherwise, the learned model may exhibit qualitatively incorrect dynamics.
Existing methods often leverage input convex neural networks (ICNNs) \cite{amos2017input} to construct the Hamiltonian function and attempt to enforce (approximately) a zero gradient of the Hamiltonian at a stable equilibrium. 
Despite their demonstrated advantages in enforcing the convexity of the Hamiltonian \cite{sanchezescalonilla2024robustneuralidapbcpassivitybased, rothStablePortHamiltonianNeural2025}, ICNN-based methods are not applicable in many real-world systems, where the Hamiltonian function is non-convex and may possess multiple isolated stable equilibria.

Addressing the above challenge, this paper presents a learning-based port-Hamiltonian model in which \emph{the stability of multiple equilibria and the port-Hamiltonian structure are preserved}.
Our main contributions are summarized below.
\begin{itemize}
    \item We propose a novel neural Hamiltonian function that ensures stability at multiple stable equilibria of the PHS, 
    overcoming a critical drawback of convex Hamiltonian functions typically assumed in existing methods.
    \item We provide a rigorous stability analysis of the proposed learning-based port-Hamiltonian model, showing that the stable equilibria are preserved. 
    \item Two numerical examples are presented to validate the accuracy of the proposed model and its effectiveness in preserving stable equilibria, and to compare its performance with an ICNN-based method.
\end{itemize}

The rest of this paper is structured as follows. 
Section~\ref{sec:problem} presents an overview of PHS and its equilibria.
Section~\ref{sec:method} introduces the proposed method and describes the learning procedure.
In Section~\ref{sec:stability}, we theoretically analyze 
the stability of the proposed model.
Next, two numerical examples are investigated in Section~\ref{sec:result}. 
Finally, Section~\ref{sec:conclusion} concludes the paper and presents future work.

\noindent{\bf Notations:}
$\Nbb$, $\Rbb$, $\Rbb_{>0}$, and $\Rbb_{\geq 0}$ are the sets of natural, real, positive real, and non-negative real numbers, respectively.
$\Vert \cdot \Vert_2$ stand for the $\ell_2$ norm. 
For two symmetric matrices with the same dimensions, $A (\succeq)\succ B$ implies that $A-B$ is a positive (semi-)definite matrix.
For a square matrix $A$, $\tr(A)$ is the sum of the elements in its main diagonal.
$\calB(\xbf, \varepsilon) = \{ \ybf \,|\, 
\Vert \ybf - \xbf \Vert_2 \leq \varepsilon\}$ is a ball centered at $\xbf$ with radius $\varepsilon > 0$.

\section{Preliminary and Problem Statement}
\label{sec:problem}

\subsection{Port-Hamiltonian systems (PHS)}
Consider a nonlinear PHS \cite{schaftPortHamiltonianSystemsTheory2014} in the following form
\begin{equation}
  \begin{cases}
    \dot \xbf = (J(\xbf)-R(\xbf))     \displaystyle
    \gradH(\xbf) + G(\xbf) \ubf, &
    \\
    \ybf = G^\top(\xbf)     \displaystyle
    \gradH(\xbf), &
  \end{cases}
  \label{eq:PHS}
\end{equation}
where $\xbf \in \Rbb^n$, $\ubf \in \Rbb^m$, and $\ybf \in \Rbb^m$ are the state, input, and output vectors, respectively.
$H: \Rbb^n \ra \Rbb_{\geq 0}$ is the Hamiltonian function representing the system's internal energy with respect to state $\xbf$.
$J(\xbf) = -J^\top(\xbf) \in \Rbb^{n\times n}$ is the structure matrix of the interconnection among the energy storage elements in the system.
$0 \preceq R(\xbf) = R^\top(\xbf) \in \Rbb^{n\times n}$ is the dissipative matrix of resistive structure representing energy losses in the system.
And $G(\xbf) \in \Rbb^{n\times m}$ is the port matrix describing how energy enters and exits the PHS.
From a system-theoretic perspective, the PHS is dissipative with respect to the supply rate $\ubf^\top \ybf$, with the Hamiltonian $H(\xbf)$ as an energy storage function
\begin{align}
\dot H(\xbf) &= 
\begin{multlined}[t]
\nabla_\xbf^\top H(\xbf) (J(\xbf)-R(\xbf)) \gradH(\xbf) \\ 
+ \nabla_\xbf^\top H(\xbf) G(\xbf) \ubf    
\end{multlined}
\nonumber
\\
&= - \nabla_\xbf^\top H(\xbf) R(\xbf) \gradH(\xbf)  + \ubf^\top \ybf,
\label{eq:energy}
\end{align}
where the last equality follows from the fact that the quadratic form of the skew-symmetric matrix $J(\xbf)$ vanishes.

\subsection{Stability and stability-preserving in PHS}

We define the stability of an equilibrium of a PHS by considering the autonomous PHS obtained from \eqref{eq:PHS} by setting $\ubf = \mathbf{0}$ as
\begin{equation}
  \begin{cases}
    \dot \xbf = (J(\xbf)-R(\xbf))     \displaystyle
    \gradH (\xbf), &
    \\
    \ybf = G^\top(\xbf)     \displaystyle
    \gradH(\xbf). &
  \end{cases}
  \label{eq:aPHS}
\end{equation}

\begin{definition}[Equilibrium point] $\xeq$ is said to be an equilibrium point of the autonomous PHS \eqref{eq:aPHS} if
$(J(\xeq)-R(\xeq))\nabla_\mathbf{x}H(\xeq) = \mathbf{0}$. 
\label{def:equilibria}
\end{definition}
\begin{definition}[Lyapunov stability] 
The autonomous system \eqref{eq:aPHS} has a (uniformly) stable equilibrium at $\xeq$ if for all $\epsilon > 0$, there exists $\delta(\epsilon) > 0$ such that
$\Vert \xbf(t_0) - \xeq \Vert_2 < \delta(\epsilon)$
implies
$\Vert \xbf(t) - \xeq \Vert_2 \leq \epsilon$, $\forall t > t_0$.
Additionally, if $\lim_{t\ra \infty} \Vert \xbf(t) -\xeq \Vert_2 = 0$, $\xeq$ is 
asymptotically stable. 
\label{def}
\end{definition}

Since $\ubf= {\bf 0}$, 
no external energy is supplied to the system and dissipativity implies that the system cannot increase its stored energy.
Therefore, the Hamiltonian satisfies the dissipation inequality $\dot{H}(\xbf) \leq 0$.
As a result, the stability of an equilibrium of \eqref{eq:aPHS} can be determined via 
Lyapunov stability theory by finding a suitable Lyapunov function.
In a PHS, the Hamiltonian function $H(\xbf)$ is often chosen as the Lyapunov function
\cite{schaftPortHamiltonianSystemsTheory2014}. 
Moreover, if the dissipation matrix satisfies $R(\xbf) \succ 0$ in a neighborhood of the equilibrium, then the equilibrium is asymptotically stable.
In addition, a PHS may possess multiple stable equilibria, for instance the chain pendulum system \cite{lee2012dynamics} and the mass–spring chain system \cite{ge2004position}.
The stability of these equilibria is generally guaranteed only locally.
Therefore, employing an ICNN for the Hamiltonian function is not suitable for capturing behavior of the PHS at all stable equilibria.

In this paper, the Hamiltonian function $H(\xbf)$ is unknown and the following assumption is adopted.

\begin{assumption}
\label{ass_hamil}
The Hamiltonian function 
\(H(\mathbf{x})\) is continuously differentiable, and 
the stable equilibria of the PHS~\eqref{eq:PHS} are known.
\end{assumption}

The stable equilibrium of the PHS \eqref{eq:PHS} can be identified from physical insights or experimental observations, as trajectories starting from a stable equilibrium return to its neighborhood under a small input perturbation.

\noindent
\textbf{Problem statement:}
This paper develops a learning-based port-Hamiltonian modeling framework that (i) preserves the intrinsic PHS structure, (ii) guarantees stability at known stable equilibria of the system, and (iii) relaxes the restrictive convexity requirement on the Hamiltonian function, while maintaining high modeling accuracy.

\section{Data-Driven Structure- and Stability-Preserving PHS}
\label{sec:method}

This section presents our method for learning a structure- and stability-preserving PHS from data.
We first state a sufficient condition for a stable equilibrium, then present our method for learning a PHS and preserving multiple stable equilibria of the system.

\subsection{Condition for a stable equilibrium}

The following results are useful for investigating the stability of equilibria.
\begin{definition}
\label{asm:fonc}
Let $H:\mathbb{R}^n \to \mathbb{R}$ be continuously differentiable. Then $\mathbf{x}_{\mathrm{eq}}$ is a strict minimum of $H$ if 
there exists $\varepsilon>0$ such that $H(\xbf) > H(\xeq)$ for all $\xbf \in \Bc(\xeq,\varepsilon)\setminus \{\xeq\}$.
Equivalently, there exists $\varepsilon^\prime >0$ and a class-$\mathcal{K}$ function $\alpha$ (see \cite[Chapter 3]{khalil2002nonlinear}) such that
  \begin{equation}
    \label{eq:strictly}
    H(\xbf) - H(\xeq) \geq \alpha(\Vert \xbf - \xeq \Vert_2), \; \forall \xbf \in \Bc(\xeq,\varepsilon^\prime)\text.
  \end{equation}
\end{definition}

\begin{lemma}
\label{lem:unique}
  Let $H:\mathbb{R}^n \to \mathbb{R}$ be continuously differentiable. If $\xeq$ is a strict minimum of $H$, then there exists $\varepsilon > 0$ such that $\xeq$ is the unique solution of $\gradH(\xbf)=\mathbf{0}$ in $\Bc(\xeq,\varepsilon)$.
\end{lemma}
\begin{proof}
See Appendix \ref{sec:proofLem1}
\end{proof}

With the help of Lemma~\ref{lem:unique}, the following lemma establishes a stability condition for an equilibrium of the autonomous PHS \eqref{eq:aPHS}.
Unlike existing studies \cite{rothStablePortHamiltonianNeural2025}, our result does not require the Hamiltonian to be convex.
\begin{lemma} \label{lem:local_stable}
  Suppose that
  $H(\xbf)$ satisfies Assumption~\ref{ass_hamil}, and
  $\xeq$ is a strict minimum of $H$.
  Then, $\xeq$ is a stable equilibrium of the system \eqref{eq:aPHS}. 
  Moreover, if $R(\xeq ) \succ 0$,
  $\xeq$ is an asymptotically stable equilibrium.
\end{lemma}
\begin{proof}
See Appendix \ref{sec:proofLem2}
\end{proof}

Using Lemma~\ref{lem:local_stable}, we can enforce the known stable equilibria via strict minima of the Hamiltonian function.
According to Definition \ref{eq:aPHS}, any strict minimum of $H$ is an equilibrium of the PHS \eqref{eq:PHS}.
Subsequently, we develop a data-driven method for learning the Hamiltonian that preserves the stability of a single stable equilibrium \( \xeq \), then extend it to the case of multiple stable equilibria.

\subsection{The proposed neural Hamiltonian \label{Hr}}

A critical component of our method is to guarantee that $H$ satisfies the stability requirement for a given stable equilibrium $\xeq$.
We consider a neural network (NN) model for learning the Hamiltonian $H$.
In conventional approaches,
 $\gradH(\xeq) = {\bf 0}$ is ensured (loosely) by adding a regularization term $\lambda \big\|\gradH(\xeq)\big\|_2^2$ to the loss function with a suitable weight $\lambda > 0$. 
However, this approach trades off the training fit for the zero-gradient constraint and is not able to guarantee exactly that $\gradH(\xeq) = {\bf 0}$.
Moreover, this approach typically does not provide theoretical guarantee for the stability of $\xeq$ in the learned model as does Lemma~\ref{lem:local_stable}.
We propose a specially designed NN model architecture for $H$ that ensures the stability at $\xeq$ by design without compromising the training fit.
In Section~\ref{sec:stability}, we will provide theoretical guarantee for local stability of strong minima of $H$ in light of Lemma~\ref{lem:local_stable}.

Our proposed model for $H$ has the form $\hat H(\xbf)  = \FNN(\xbf) h(\Vert \xbf - \xeq \Vert_2)$, where $\FNN(\cdot)$ is a neural network taking an input in $\Rbb^n$ 
and returning an output in $\Rbb_{>0}$ (\eg by using a positive activation function in the output layer), and $h(\cdot)$ is a predefined function from 
$\Rbb_{\geq 0}$ to $\Rbb_{\geq 0}$.
The function $h(\cdot)$ must meet several criteria.
First, $h(\Vert \xbf - \xeq \Vert_2) = 1$ whenever $\Vert \xbf - \xeq \Vert_2 \geq b$ for a small $b$, 
so that it does not affect the prediction of $\FNN(\cdot)$.
Second, it must ensure that $\nabla_{\!\xbf} \hat{H}(\xeq) = 0$.
Since 
\begin{multline*}
\nabla_{\!\xbf} \hat{H}(\xbf) = \nabla_\xbf \FNN(\xbf) h(\Vert \xbf - \xeq \Vert_2)
\\
+ \FNN(\xbf)  \frac{\partial h(\Vert \xbf - \xeq \Vert_2)}{\partial \xbf},
\end{multline*}
to make $\nabla_{\!\xbf} \hat{H}(\xeq) = {\bf 0}$, we design $h(\cdot)$ such that $h (0) =  0$ and $h^\prime (0) =  0$ {because
$\frac{\partial h}{\partial \xbf} = h^\prime(\Vert \xbf - \xeq \Vert_2) 
\frac{\xbf - \xeq}{\Vert \xbf - \xeq \Vert_2}$.
}
However, $\frac{\xbf - \xeq}{\Vert \xbf - \xeq \Vert_2}$ is not continuous at $\xbf = \xeq$.
Therefore, we make a slight change as
\begin{align}
  \hat{H}(\xbf)  = \FNN(\xbf) h(\sigma(\Vert \xbf - \xeq \Vert_2)),
  \label{nnHamil}
\end{align}
where $\sigma: \Rbb_{\geq 0} \ra \Rbb_{\geq 0}$ is such that $\sigma(0) = 0, \frac{\pa\sigma}{\pa z}(0) = 0$ and $\sigma$ is continuously differentiable in $\Rbb_{\geq 0}$.
Several candidates for $\sigma(z)$ are available, for example, we choose $\sigma(\Vert \xbf - \xeq\Vert_2) = \sqrt{\Vert \xbf - \xeq\Vert_2^2 + \delta^2} - \delta$ (for $\delta > 0$), whose derivative $\frac{\pa \sigma}{\pa \xbf } = \frac{\xbf - \xeq}{\sqrt{\Vert \xbf - \xeq\Vert_2^2 + \delta^2}}$ goes to ${\bf 0}$ as $\xbf \ra \xeq$.
The function $h(\cdot)$ is chosen as the $d$-differentiable step function \cite{nguyen2020distributed} given by
\begin{align}
    h(\sigma) = \begin{cases}
        0, \qquad\qquad\text{if}~\sigma\in (-\infty, 0],
        \\
        \Big(\frac{\sigma}{\sigma_b}\Big)^{d+1} \sum_{j=0}^{d} 
            {d+j \choose j} {2d+1 \choose d-j} \Big( \frac{-\sigma}{\sigma_b} \Big)^j
            \\
            \qquad\qquad\qquad\qquad\text{if}~\sigma\in (0,\, \sigma_b),
        \\
        1, \qquad\qquad\text{if}~\sigma\in [\sigma_b, +\infty)\text,
    \end{cases}
    \label{eq:sms}
\end{align}
where $\sigma_b = \sigma(b)$ and $b>0$. 
We then have
\begin{equation}
  \label{diff_Hr}
  \nabla_{\!\xbf} \hat{H}(\xbf) \!=\!  \nabla_\xbf \FNN(\xbf) h(\sigma) 
  \!+\! \FNN(\xbf)  \frac{h^\prime(\sigma) (\xbf \!-\! \xeq)}{\sqrt{\Vert \xbf \!-\! \xeq \Vert_2^2 \!+\! \delta^2}}.
\end{equation}
From \eqref{eq:sms}, we have $h'(\sigma) = 0$ for $\sigma \geq \sigma_b$, and for $0< \sigma < \sigma_b$, $h^\prime(\sigma) = g(\sigma) \frac{h(\sigma)}{\sigma}$, where
\begin{equation}
  \label{eq:g-in-h-prime}
  g(\sigma) = \frac{\sum_{j=0}^{d} (-1)^j (d+j+1) {d+j \choose j} {2d+1 \choose d-j} \left( \frac{\sigma}{\sigma_b} \right)^j}{  \sum_{j=0}^{d} (-1)^j {d+j \choose j} {2d+1 \choose d-j} \left( \frac{\sigma}{\sigma_b} \right)^j}\text.
\end{equation}
An illustration of $h(\sigma(\Vert \xbf \Vert_2))$ when $d=1$ and $\xeq = {\bf 0}$ is shown in Fig.~\ref{hx}.
Here, Fig.~\ref{hx-1} and Fig.~\ref{hx-2} show that $h$ is continuously differentiable.  
In addition, for $\xbf$ such that $\Vert \xbf - \xeq\Vert_2 \geq b$, $h(\sigma) = 1$ and $h^\prime(\sigma) = 0$, thus
$\hat{H}(\xbf) = \FNN (\xbf)$ and $\nabla_{\!\xbf} \hat{H}(\xbf)\ = \nabla_\xbf \FNN (\xbf)$.
The effect of $h(\cdot)$ is activated only in the chosen region $\Omega_b = \calB(\xeq, b)$.

\begin{remark}
The choice of $\sigma(\Vert \xbf - \xeq \Vert_2)$ yields the term
$ \frac{\xbf - \xeq}{\sqrt{\Vert \xbf - \xeq \Vert_2^2 + \delta^2}}$
in $\frac{\partial h}{\partial \xbf}$, whose magnitude is always less than $1$ and approaches $1$ when $\delta$ is sufficiently small. Hence, the contribution of $\frac{\partial h}{\partial \xbf}$ to $\nabla_{\xbf}\hat H$ is limited.
If we instead chose $h(\Vert \xbf - \xeq \Vert_2^2)$ to avoid the non-differentiability of $\frac{\xbf - \xeq}{\Vert \xbf -\xeq \Vert}$ at $\xeq$, then $\xbf - \xeq$ would appear in
$\frac{\partial h}{\partial \xbf}
= 2 h^\prime(\Vert \xbf - \xeq \Vert_2^2)(\xbf - \xeq)$.
In this case, it would be difficult to design $h(\cdot)$ 
to mitigate the effect of the magnitude of $\xbf - \xeq$ in $\nabla_{\xbf}\hat H$.
\end{remark}

\begin{remark}
\label{rm:relax}
Based on the design of $\hat H(\xbf)$ in \eqref{nnHamil}, we may consider the following more general form: $
\hat H(\xbf)
= \FNN(\xbf)\, h\!\left(\sigma(\|\xbf-\xeq\|_2)\right)
+ w(\xbf)\bigl(1-h\!\left(\sigma(\|\xbf-\xeq\|_2)\right)\bigr)$,
where $w(\xbf)$ denotes a learnable differentiable function such that $w(\xbf)\ge 0$ for all $\xbf$, and $\nabla_{\xbf} w(\xeq)=\mathbf{0}$.
The function $w(\xbf)$ provides a relaxation of \eqref{nnHamil} within $\Omega_b$ and is independent of $\FNN(\xbf)$. Specifically, define $
g(\xbf) = w(\xbf)\bigl(1 - h\!\left(\sigma(\|\xbf - \xeq\|_2)\right)\bigr)$,
whose effect on $\hat H(\xbf)$ is restricted to the region $\Omega_b$. Indeed, $g(\xbf) = 0$ and $\nabla_{\xbf} g(\xbf) = \mathbf{0}$ for all $\xbf \notin \Omega_b$, implying that $g(\xbf)$ has no influence outside $\Omega_b$.
\end{remark}

\begin{figure}[!tb]
  \subfloat[$h(\sigma)$ w.r.t. $\xbf$ ($n=1$). \label{hx-1}]{\includegraphics[width = 0.5\linewidth]{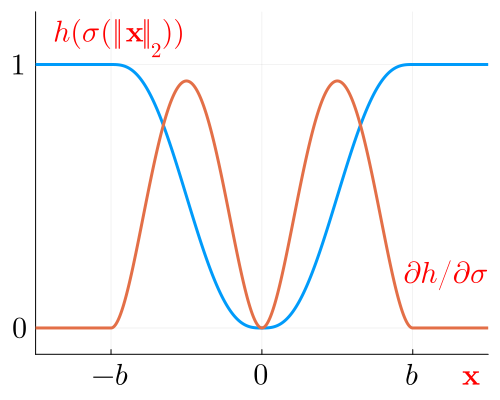}}
  \subfloat[$h(\sigma)$ w.r.t. $\xbf $ ($n=2$). \label{hx-2}]{\includegraphics[width = 0.5\linewidth]{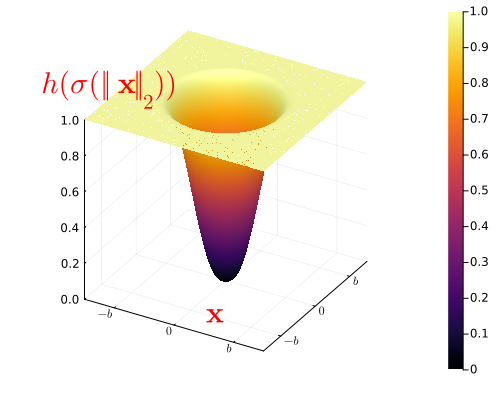}}
  \centering
  \caption{Illustration of $h(\sigma(\Vert \xbf \Vert_2))$ at $\xeq = {\bf 0}$. \label{hx}}
\end{figure}

\subsection{Learning Port-Hamiltonian model}

Besides the Hamiltonian $H(\xbf)$, the matrices $J(\xbf)$, $R(\xbf)$, and $G(\xbf)$ may be unknown.
If the parametric structures of the matrices $J(\xbf)$, $R(\xbf)$, and $G(\xbf)$ are 
known, but the parameters themselves are unknown, we can learn these parameters along with the parameters of the neural Hamiltonian $\hat H$.
If $J(\xbf)$, $R(\xbf)$, and $G(\xbf)$ are completely unknown, 
the decompositions
$J(\xbf) =  T_J(\xbf) - T_J^\top(\xbf)$ and $R(\xbf) = T_R(\xbf) T_R^\top(\xbf)$, where 
$T_J$ and $T_R$ are lower triangular matrices, are used to preserve skew-symmetric and symmetric structures of $J$ and $R$.
As a result, we can use neural networks $\mathbf{NN}_{J}(\xbf), \mathbf{NN}_{R}(\xbf)$, and $\mathbf{NN}_{G}(\xbf)$, whose outputs correspond to unknown elements of $T_J(\xbf)$, $T_R(\xbf)$, and $G(\xbf)$, respectively.
Accordingly, the learning model of the PHS is given by
\begin{equation}
  \begin{cases}
    \dot \xbf = (\hat{J}(\xbf)-\hat{R}(\xbf))     \displaystyle
    \nabla_{\!\xbf} \hat{H}(\xbf) + \hat{G}(\xbf) \ubf, &
    \\
    \ybf = \hat{G}^\top(\xbf)    
    \nabla_{\!\xbf} \hat{H}(\xbf), &
  \end{cases}
  \label{eq:lPHS}
\end{equation}
where $\hat{J}$, $\hat{R}$, and $\hat{G}$ are learned matrices.

For training the PH model \eqref{eq:lPHS},
let us assume that a training dataset $\Dc = \{\xbf_k, \ubf_k, t_k\}_{k=0,1,\dots,N_d}$ of the PHS \eqref{eq:PHS}, where $N_d$ is the number of samples, is given.
We aim to learn the dynamics \eqref{eq:PHS}, in particular $J$, $R$, $G$, and the Hamiltonian $H$, that preserve the PH structure and the stable equilibrium $\xeq$.
Hereafter, $\xbf(t)$ (or $\xbf$) denotes the state variable obtained from the learned PH model \eqref{eq:lPHS}, whereas $\mathbf{x}_k$ represents the state of the PHS \eqref{eq:PHS} at time $t_k$.

Let $\bm{\theta}$ be the vector of all parameters of $\hat H(\xbf)$, $\hat J(\xbf)$, $\hat R(\xbf)$, and $\hat G(\xbf)$ of the 
learning model \eqref{eq:lPHS}.
Learning $\bs\theta$ from the dataset $\calD$ can be carried out by minimizing the loss function
\begin{equation}
  \label{eq:loss-general}
  \Lc = \frac{1}{N_d}\sum_{k=0}^{N_d-1} \Vert \hat\xbf_{k+1} - \xbf_{k+1} \Vert_2^2 +  \lambda \Vert \bm\theta \Vert_2^2,
\end{equation}
where $\lambda > 0$ is the regularization weight and
$\hat\xbf_{k+1} = \xbf(t_{k+1})$ is  
obtained by integrating \eqref{eq:lPHS} from $\xbf_k$ at time $t_k$ to time $t_{k+1}$.
Here, we employ symplectic integrators, such as the symplectic Euler method and the Verlet method \cite{kotyczka2019discrete}, to solve the Hamiltonian differential equation. 

\subsection{Preserving multiple stable equilibria}
\label{sec:mul_equi}
 
We extend the proposed neural Hamiltonian approach to the case of multiple isolated stable equilibria.
Consider the PHS \eqref{eq:PHS}, which admits  $n_\mathrm{eq}$ isolated stable equilibria $\big\{\xeq^{(i)}\big\}_{i = 1,\dots, n_\mathrm{eq}}$ characterized by  $\gradH(\xeq^{(i)}) = {\bf 0}$. The proposed Hamiltonian for multiple isolated equilibira is
\begin{align}
\hat H(\xbf) = \FNN(\xbf) h_{\Sigma}(\xbf; \xbf^{(1)},\dots, \xeq^{(n_\mathrm{eq})}),
\label{eq:Hhat_mul}
\end{align}
where 
$h_{\Sigma} = \sum_{i=1}^{n_\mathrm{eq}} \! h_i(\sigma_i) - n_\mathrm{eq} + 1$, $\sigma_i = \sigma(\Vert \xbf - \xeq^{(i)}\Vert_2)$, and
$h_i(\sigma_i)$ is the step function \eqref{eq:sms} with parameter $b = b_i$ standing for the active region of $h_i$ as in Fig. \ref{hx}.
We select the parameters $b_i$ such that there must be no overlap between the active regions $\Omega_{b_i} = 
\calB(\xeq^{(i)}, b_i)$, that is,  $\Omega_{b_i} \cap \Omega_{b_j} = \emptyset$ for all $i,j = 1,\dots, n_\mathrm{eq}$ and $i \neq j$.
Note that, for all isolated equilibria, the parameters $b_i$ can be chosen sufficiently small to satisfy the above requirement.

Based on this setting, we can analyze individually the stability of each $\xeq^{(i)}$ in the learned PH model since $h_i(\sigma(\Vert \xbf - \xeq^{(i)}\Vert_2)) = 1$ over $\xbf \in \Omega_{b_j} (i\neq j)$.
From \eqref{eq:Hhat_mul}, we have
\begin{multline}
\nabla_\xbf {\hat H}(\xbf) =  \nabla_\xbf \FNN(\xbf) \left(\sum_{i=1}^{n_\mathrm{eq}} h_i(\sigma_i) \!-\! n_\mathrm{eq} \!+\! 1 \right)
\\
+ \FNN(\xbf) \sum_{i=1}^{n_\mathrm{eq}}  \frac{h_i^\prime(\sigma_i)(\xbf \!-\! \xeq^{(i)})}{\sqrt{\Vert \xbf \!-\! \xeq^{(i)} \Vert_2^2 \!+\! \delta^2}}.
\label{eq:dHhat_mul}
\end{multline}
It follows from the definitions of $h_i$ and $\Omega_{b_i}$ that, for all $i = 1,\dots, n_\mathrm{eq}$,
$h_{\Sigma}(\xeq^{(i)}; \xbf^{(1)},\dots, \xeq^{(n_\mathrm{eq})}) = 0$ and
$h_i^\prime(\sigma_i) = 0$ at $\xbf = \xeq^{(i)}$.
Moreover, for all $\xbf \in \Om_{b_i}$,
$h_j(\sigma_j) = 1$ and $h_j^\prime(\sigma_j) = 0$ for all $j\neq i$.
As a result, for all $\xbf \in \Om_{b_i}$, \eqref{eq:dHhat_mul} becomes
\begin{multline*}
    \nabla_\xbf {\hat H}(\xbf) = \nabla_\xbf \FNN(\xbf)  h_i(\sigma_i) 
    \\
+ \FNN(\xbf)   \frac{h_i^\prime(\sigma_i)(\xbf - \xeq^{(i)})}{\sqrt{\Vert \xbf - \xeq^{(i)} \Vert_2^2 + \delta^2}}\text, 
\end{multline*}
which has the same form as \eqref{diff_Hr}. 
Thus, the stability of the learned PH model with multiple equilibria can be analyzed by examining the stability of each individual equilibrium.

\section{Stability analysis}
\label{sec:stability}
This section presents the stability-preservation analysis of the learned PH model.

\begin{theorem} \label{lem_stability}
  With the proposed neural Hamiltonian function \eqref{nnHamil}, the learned PH model 
  \eqref{eq:lPHS}  has a stable equilibrium at $\xeq$.
  Moreover, if $\hat R(\xeq) \succ 0$, $\xeq$ is an asymptotically stable equilibrium of \eqref{eq:lPHS}.
\end{theorem}

\begin{proof}
  The PH model \eqref{eq:lPHS} has an 
  equilibrium at $\xeq$ because, by the definition of $\hat{H}(\xbf)$ in \eqref{nnHamil} and its derivative in \eqref{diff_Hr}, $\hat{H}({\xeq}) = 0$ and $\nabla \hat{H}({\xeq}) = \mathbf{0}$.
  Since $\FNN (\xbf) > 0$, there exist $\varepsilon$ such that 
  $\hat{H}(\xbf) > 0 = \hat{H}({\xeq})$
  for all $\xbf \in B(\xeq, \varepsilon)$.
  Thus, $\xeq$ is a strict minimum of $\hat{H}(\xbf)$.
  By Lemma~\ref{lem:local_stable}, $\xeq$ is a stable equilibrium of \eqref{eq:lPHS} and, if $\hat R(\xeq) \succ 0$, $\xeq$ is asymptotically stable.
\end{proof}

Based on Theorem \ref{lem_stability}, the following proposition establishes stability preservation for multiple stable equilibria of the learned PH model.

\begin{proposition}
The proposed neural Hamiltonian function \eqref{eq:Hhat_mul} ensures that the learned PH model \eqref{eq:lPHS} has multiple stable equilibria at 
$\big\{\xeq^{(i)}\big\}_{i = 1,\dots, n_\mathrm{eq}}$.
In addition, if $\hat R(\xeq^{(i)}) \succ 0$, the equilibrium $\xeq^{(i)}$ is asymptotically stable.
\end{proposition}
\begin{proof}
The proof follows a similar argument as in the proof of Theorem~\ref{lem_stability}, showing that $\xeq^{(i)}$ is a strict minimum of $\hat H(\xbf)$ for all $i = 1, \dots, n_\mathrm{eq}$.
\end{proof}

Although Theorem \ref{lem_stability} confirms that $\xeq$ is a stable equilibrium of \eqref{eq:lPHS}, an estimation of the region of attraction still remains.
Consider a fully connected neural network of $L$ layers $\FNN(\xbf) = f_L(W_d(f_{L-1}(W_{L-1}(\dots W_1 \xbf + b_1 \dots) + b_{L-1}) + b_L))$, 
where $W_i$ is a weight matrix at layer $i$ with bias vector $b_i$, and 
$f_i$ is a continuous differentiable activation function for $i = 1,\dots, L$ such that $f_L > 0$.
Define the set
\begin{align} \label{bound}
  \Rc: \left\{\xbf  \in \Om_b \big|
  \Vert \nabla \FNN \Vert_\infty < \beta(\Delta \xbf) \right\},
\end{align}
where $ \beta(\Delta \xbf) =  \frac{c_L \ldotp h^\prime(\sigma) \Vert \Delta \xbf \Vert_\infty }{h(\sigma)\sqrt{\Vert \Delta  \xbf \Vert_2^2 + \delta^2}}$, $c_L = \min_{\xbf} \FNN (\xbf)$, and $\Delta\xbf = \xbf - \xeq$.  
Note that $\beta(\Delta\xbf) > 0$ and is well-defined in $\Omega_b = \calB(\xeq, b)$ except at $\Delta\xbf = \mathbf{0}$.
Recall that $\frac{h^\prime(\sigma)}{h(\sigma)} = \frac{g(\sigma)}{\sigma}$ with $g(\sigma)$ defined in \eqref{eq:g-in-h-prime}.
Also note that $\lim_{\sigma \to  0} g(\sigma) = g(0)$ exists and is positive.
We consider the limit, 
\begin{multline*}
\lim_{\Vert \Delta\xbf \Vert_2 \to 0} \beta(\Delta \xbf) = 
\frac{c_L g(0)}{\delta} \lim_{\Vert \Delta \xbf \Vert_2 \to 0} 
\frac{\Vert \Delta\xbf \Vert_\infty}{\Vert \Delta \xbf \Vert_2} \frac{\Vert \Delta \xbf \Vert_2}{\sigma(\Vert \Delta\xbf\Vert_2)}.
\end{multline*}
Due to equivalence of norms, we can split the above limit and obtain
$\lim_{\Vert \Delta\xbf \Vert_2 \to 0} \beta(\Delta\xbf) = c_\beta \lim_{\Vert \Delta\xbf \Vert_2 \to 0} 
\frac{\Vert \Delta\xbf \Vert_2}{\sigma(\Vert \Delta\xbf \Vert_2)} = +\infty$ (by applying L'hopital rule for the last term), where $c_\beta$ is a positive scalar.
Hence, the set $\Rc$ is nonempty and admits $\xeq$ as an interior point since 
$\nabla\mathbf{NN_H}$ is a continuous function.
Let us consider $\Hc_{\Rc} \subseteq \Rc$ being the largest level set of $\hat H$ in $\Rc$.
\begin{proposition} \label{prop:ROA}
  If the learned PH model \eqref{eq:lPHS} has an asymptotically stable equilibrium at $\xeq$, $\Hc_{\Rc}$ is a forward invariant set such that all solutions $\xbf(t)$ of \eqref{eq:lPHS} starting from $\Hc_{\Rc}$ will asymptotically converge to $\xeq$.
\end{proposition}
\begin{proof}
  Recall the forward invariant set
  $\Omega = \{\xbf \in \Rbb^n | \frac{\pa \hat H}{\pa \xbf} = \mathbf{0}\}$.
  By the definition of $\hat H$, $\xeq \in \Omega$.
  We will show that $\frac{\pa \hat H}{\pa \xbf} = \mathbf{0}$ has only one solution $\xeq$ in $\Hc_{\Rc}$.
  The definition of $\Rc$ leads to, for all $\xbf \in \Rc$ and $\xbf \neq \xeq$,
  \[
  \FNN(\xbf)  \frac{h^\prime(\sigma) \Vert \Delta\xbf \Vert_\infty }{\sqrt{\Vert \Delta\xbf \Vert_2^2 + \delta^2}} 
  > \left\Vert \frac{\pa \FNN}{\pa \xbf} \right\Vert_\infty h(\sigma)\text.
  \]
  Then
  $\FNN(\xbf)  \frac{h^\prime(\sigma) \Delta\xbf}{\sqrt{\Vert \Delta\xbf \Vert_2^2 + \delta^2}} 
  \neq -\frac{\pa \FNN}{\pa \xbf} h(\sigma)$.
  By \eqref{diff_Hr}, $\frac{\pa H}{\pa \xbf} \neq \mathbf{0}$ for all $\xbf \neq \xeq$.
  Thus, $\frac{\pa \hat H}{\pa \xbf} = \mathbf{0}$ admits a unique solution $\xbf = \xeq$ over $\Rc$.
  Because $\Hc_{\Rc}$ is a level set of $\hat H$ and $\xeq$ is an asymptotically stable equilibrium of \eqref{eq:lPHS}, every solution of $\eqref{eq:lPHS}$ starting from $\Hc_{\Rc}$ will remain in $\Hc_{\Rc}$ and asymptotically converge to $\xeq$.
\end{proof}

\begin{remark}
Although there possibly are other equilibria in the PHS, Proposition \ref{prop:ROA} shows that there is only one equilibrium in $\Hc_\Rc$.
Therefore, Proposition \ref{prop:ROA} allows us to design a region of attraction, 
 which preserves local asymptotic stability of the full-state system at $\xeq$.  
\end{remark}

\begin{remark}
  Contrary to methods employing ICNNs as described in \cite{rothStablePortHamiltonianNeural2025}, which ensure global stability of an equilibrium, this paper proposes a local approach aimed at preserving stability around a specific equilibrium. 
  This local method offers several potential advantages. 
  First, it accommodates scenarios where the equilibrium of interest may not be globally stable; imposing global stability in such scenarios could lead to overly conservative approximations when fitting globally stable equilibria in the learned model to locally stable equilibria present in the original system.
  Second, the local approach facilitates modeling systems exhibiting multiple equilibria using a single neural network, thus eliminating the complexity associated with integrating multiple local models. 
  Finally, by judiciously defining the region $\Rc$ as indicated in \eqref{bound}, it is possible to estimate 
  regions of attraction corresponding to a desired equilibrium.
\end{remark}

\section{Numerical examples}
\label{sec:result}

This section demonstrates the effectiveness and accuracy of the proposed method through two case studies, namely a Toda lattice system and a double pendulum system.
We compare our approach with PH-ICNN \cite{rothStablePortHamiltonianNeural2025}, a method based on input convex neural networks (ICNNs) for learning PH systems from data.
The code and supplementary materials for this paper are available at:
\url{https://github.com/EasternBoy/Stability-Preserving-PHS.git}

\subsection{Toda lattice system}

We consider the nonlinear Toda lattice 
model \cite{schwerdtnerAdaptiveSamplingStructurePreserving2021, gengDataDrivenReducedOrderModels2025}, that describes the motion of a chain of elements, each connected to its nearest neighbors with \textit{exponential springs} and dampers. 
The equations of motion for the $\ell$-particle Toda lattice with such interactions can be written in the form of a PHS as in \eqref{eq:PHS} with the corresponding matrices:
\begin{align*}
&J = \begin{bmatrix}
\bm{0} & I_\ell\\
-I_\ell & \bm{0}
\end{bmatrix},~
R = \begin{bmatrix}
\bm{0} & \bm{0}\\
\bm{0} & \mathrm{diag}(\gamma_1, \dots, \gamma_\ell)
\end{bmatrix} \in \Rbb^{2\ell \times 2\ell},
\\
&G = \begin{bmatrix}
\bm{0}\\
\mathbf{e}_1
\end{bmatrix} \in \Rbb^{2\ell},~
\text{and}~
\mathbf{e}_1 = [1,0,\dots,0]^\top \in \Rbb^\ell,
\end{align*}
where 
$I_\ell$ is the identity matrix in $\Rbb^{ \ell \times \ell}$ and $\gamma_i$ is the unknown damping coefficient related to the $i$-th particle in the system.
The states are $\mathbf{x} = [\mathbf{q}^\top, \pbf^\top]^\top$, where
$\mathbf{q} = [q_1, \dots, q_\ell]^\top$ and
$\mathbf{p} = [p_1, \dots, p_\ell]^\top$.
Here, $q_i$ and $p_i$ are the displacement of the $i$-th particle from its equilibrium position and the momentum, respectively.
The Hamiltonian of the Toda lattice system is the nonlinear and non-convex function
\begin{equation*}
H(\mathbf{x}) = \sum_{i=1}^{\ell} \frac{1}{2} p_i^2 + \sum_{i=1}^{\ell-1} 
e^{q_i - q_{i+1}} + e^{q_\ell} - q_1 - \ell + \epsilon(1-\cos(q_1)),
\end{equation*}
with $\epsilon > 0$.
Accordingly, $\left(\frac{\pa H}{\pa \xbf}\right)^\top = [e^{q_1 - q_2} - 1 + \epsilon \sin(q_1), e^{q_2 - q_3} - e^{q_1 - q_2},\dots,$ $e^{q_\ell} -  e^{q_{\ell-1} - q_\ell}, \pbf^\top ] = \bm{0}$ if and only if $\pbf = \bm{0}$ and $\mathbf{q} = \bm{0}$, thus the Toda lattice system has only one stable equilibrium at $\xeq = {\bf 0}$.
Assumption \ref{ass_hamil} is then verified by the above $H(\xbf)$ and $\frac{\pa H}{\pa \xbf}$.

\begin{figure*}
\subfloat[Hamiltonian\label{fig:toda_lattice_sin:a}]{\includegraphics[width=0.33 \linewidth]{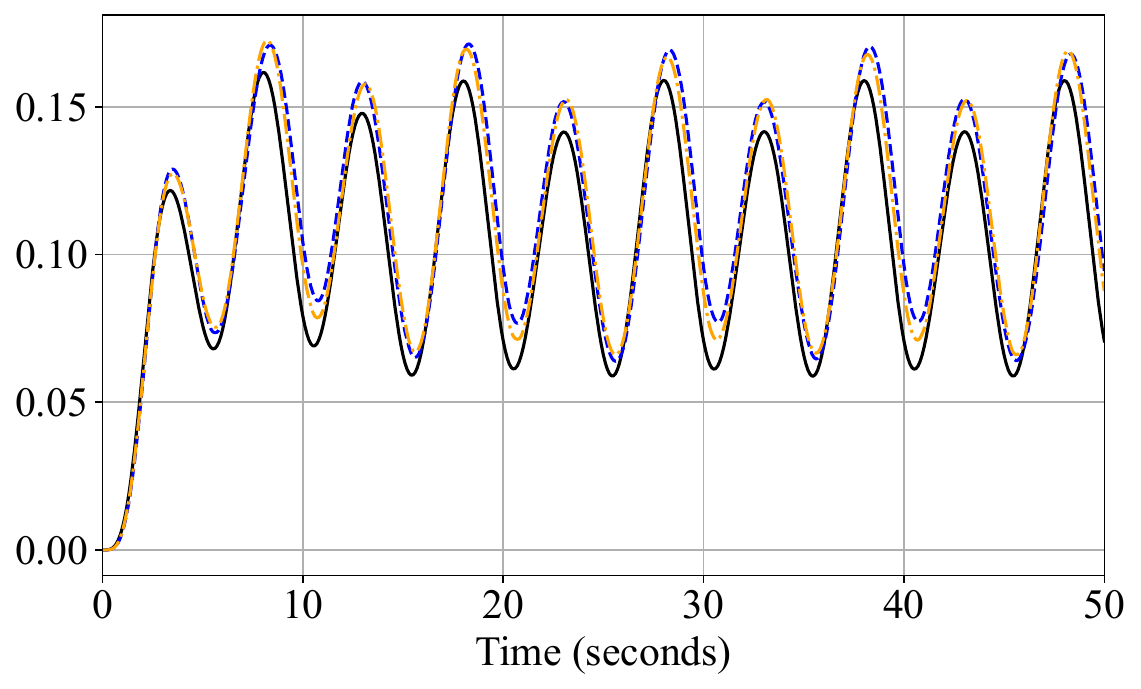}}
\subfloat[Output\label{fig:toda_lattice_sin:b}]{\includegraphics[width=0.33 \linewidth]{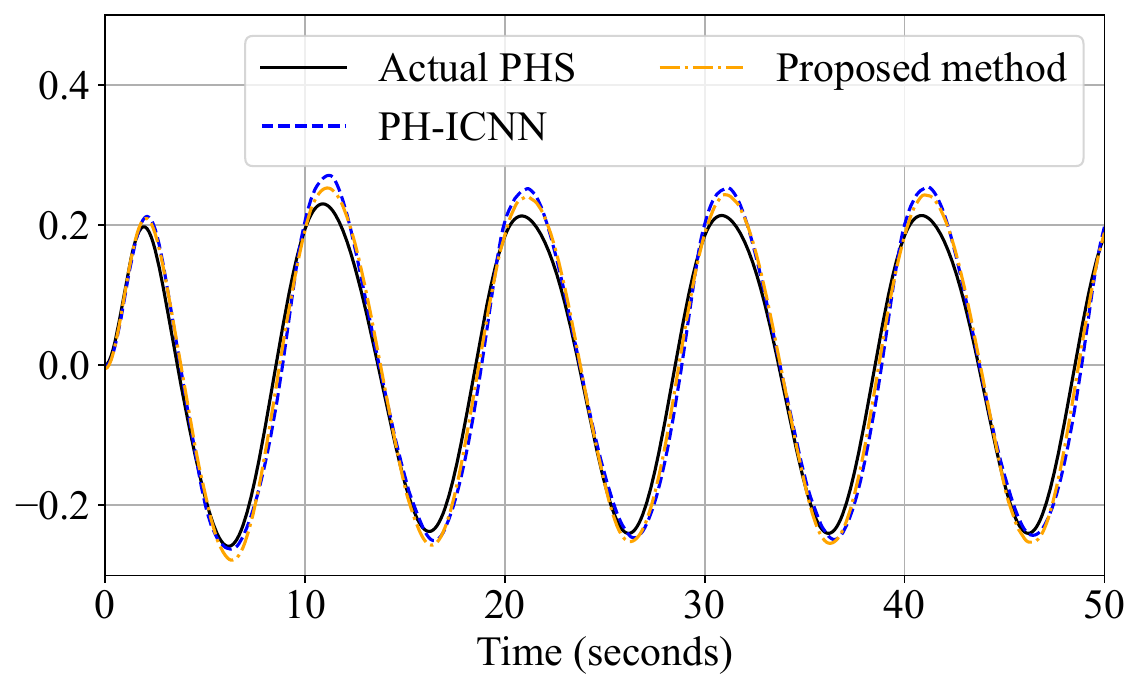}}
\subfloat[State magnitude $\Vert \xbf \Vert_2$\label{fig:toda_lattice_sin:c}]{\includegraphics[width=0.33 \linewidth]{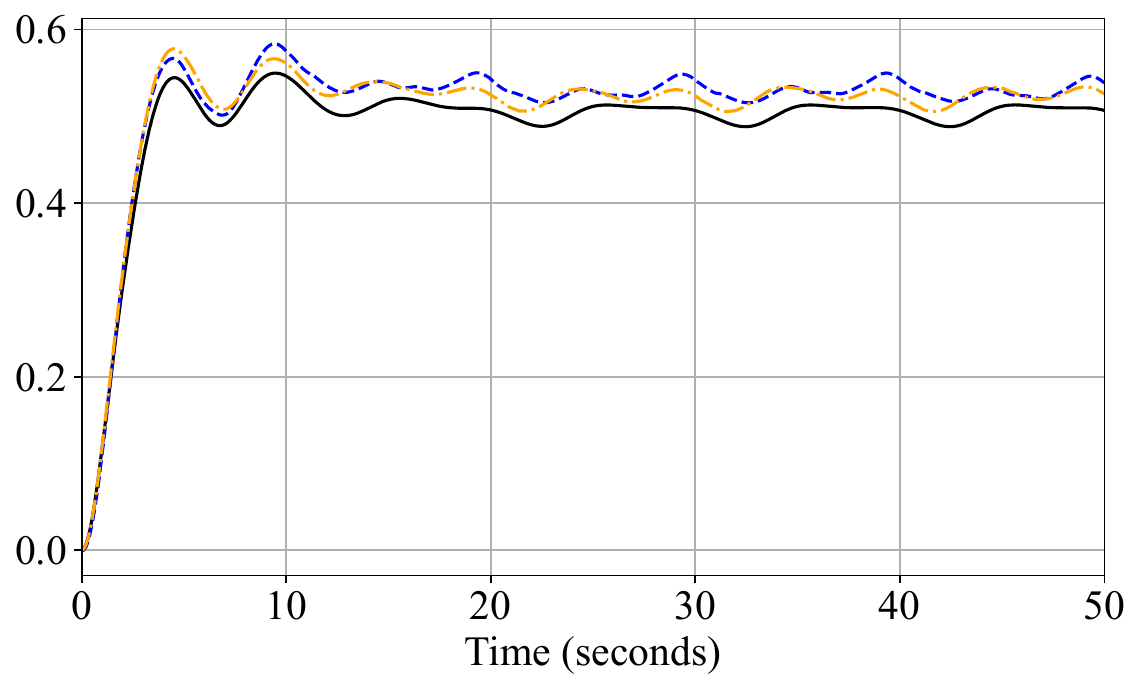}}
\caption{Responses of the actual PHS, PH-ICNN, and proposed models 
to the sinusoid test signal for the Toda lattice system.\label{fig:toda_lattice_sin}}
\end{figure*}

\begin{figure*}
\subfloat[Hamiltonian\label{fig:toda_lattice_pulse:a}]{\includegraphics[width=0.33 \linewidth]{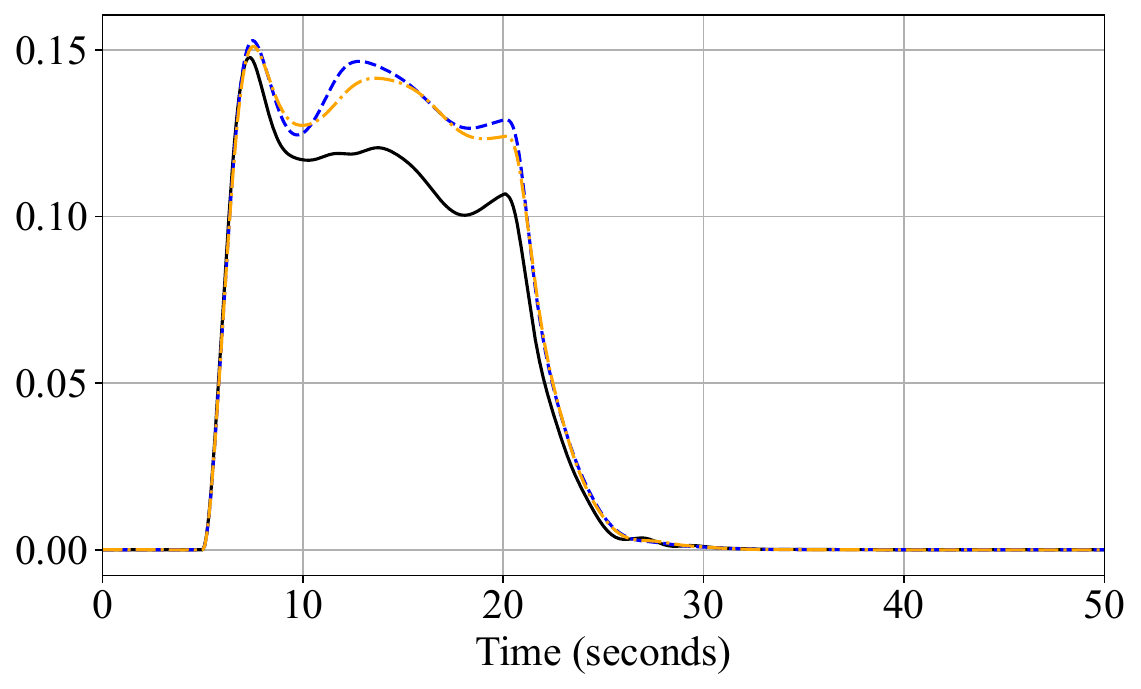}}
\subfloat[Output\label{fig:toda_lattice_pulse:b}]{\includegraphics[width=0.33 \linewidth]{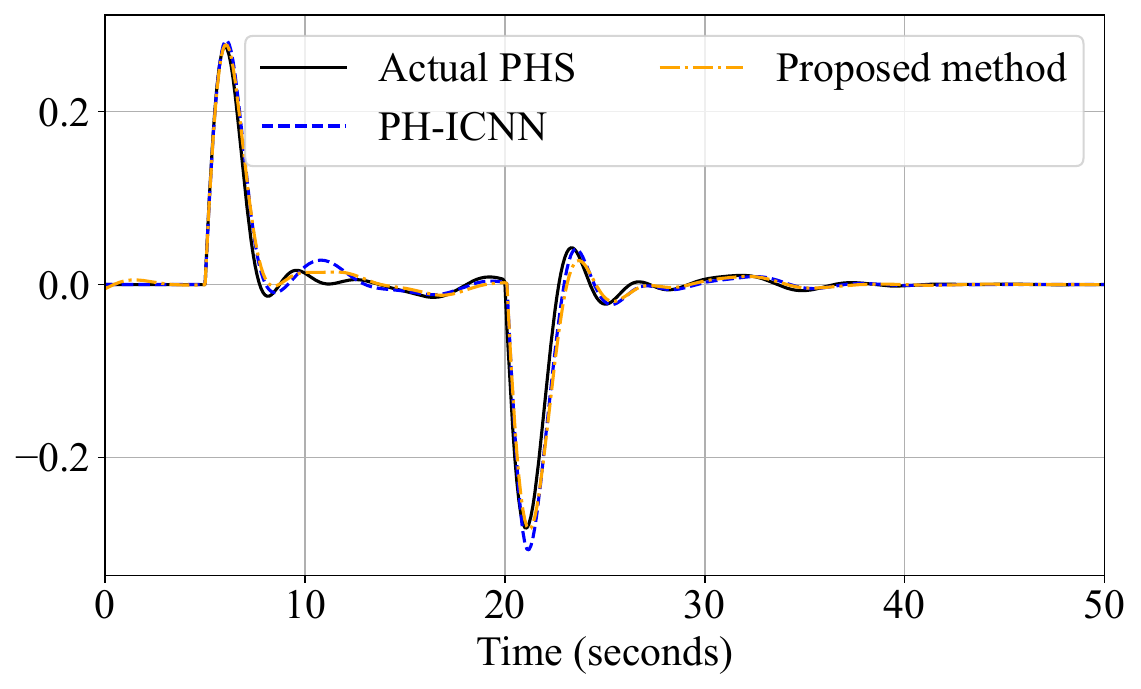}}
\subfloat[State magnitude $\Vert \xbf \Vert_2$\label{fig:toda_lattice_pulse:c}]{\includegraphics[width=0.33 \linewidth]{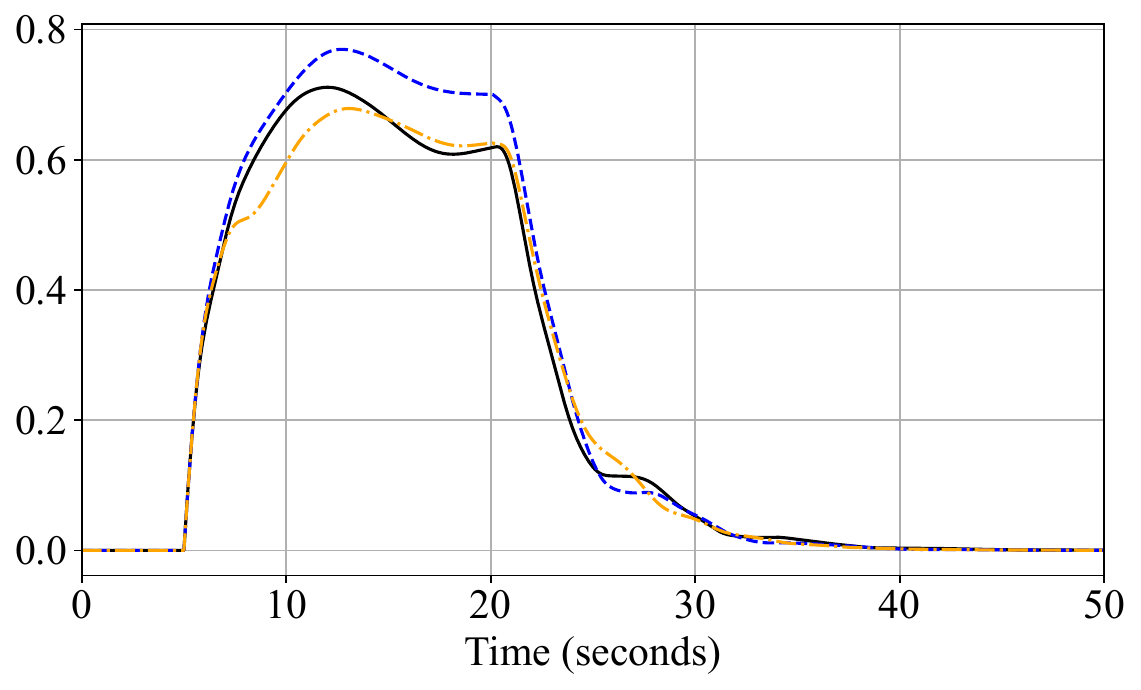}}
\caption{Responses of the actual PHS, PH-ICNN, and proposed models 
to the pulse test signal for the Toda lattice system.
\label{fig:toda_lattice_pulse}
}
\end{figure*}

In our experiment, the system dimension 
is set to \( \ell = 5 \), with system parameters $\gamma$ uniformly defined as \( \gamma_1 = \dots = \gamma_\ell = 0.5 \). 
The initial condition for the simulation is chosen as 
\( \mathbf{x}_0 = 0 \) and dimension $n =10$.
Training data was collected over the interval 
\( [0, 1000] \) seconds, with a sampling interval of \qty{0.1}{\second}. 
The system was excited by the input \( \mathbf{u}(t) = \mathbf{u}(t_k) = \text{rand}() \) for each interval \( t \in [t_k, t_{k+1}) \), where 
\( \text{rand}() \) represents a uniformly distributed random number in \([-1, 1]\). 
Numerical solutions to the differential equations were obtained using the Euler method.

In constructing the proposed Hamiltonian, the smooth step function $\sigma(\cdot)$ is set with order $d=2$ and radius $b=1$.
Two NNs are used to represent $\hat H(\xbf)$ as in Remark \ref{rm:relax}, each with four layers (two hidden layers of 32 neurons).
We evaluate the accuracy of the proposed PH model and compare it with PH-ICNN using two input signals. 
The first is a pulse signal of magnitude $0.5$ applied from \qty{5}{\second} to \qty{20}{\second}, and the second is a sinusoidal signal of magnitude $0.5$ with a period of \qty{10}{\second}. 
The initial state is set to $\mathbf{x}_\mathrm{eq} = \mathbf{0}$.

Fig. \ref{fig:toda_lattice_sin} shows the responses of PH-ICNN and the proposed learned PH model to the sinusoidal test signal, including the Hamiltonian, output, and state magnitude.
As shown in Figs.~\ref{fig:toda_lattice_sin:a} and \ref{fig:toda_lattice_sin:b}, the proposed model can accurately capture the responses of the Toda lattice system under the sinusoidal input. The Hamiltonian and output exhibit clear sinusoidal patterns. Furthermore, the results of the proposed learned PH model are more accurate (i.e., closer to the actual PHS) than those of the PH-ICNN model, as illustrated in Figs.~\ref{fig:toda_lattice_sin:a}--\ref{fig:toda_lattice_sin:c}.

Additional numerical results for the pulse test signal are presented in Fig.~\ref{fig:toda_lattice_pulse}.
The purpose of these results is to verify the preservation of stability at the origin equilibrium of the PHS in the proposed model.
From Figs.~\ref{fig:toda_lattice_pulse:a}--\ref{fig:toda_lattice_pulse:c}, the actual PHS, PH-ICNN, and the proposed model remain at the origin equilibrium until \qty{5}{s}, after which they are excited by the pulse input.
The Hamiltonian, output, and state are perturbed away from the origin and subsequently return to it, confirming the stability of the origin equilibrium in the proposed model.
Furthermore, the proposed model achieves more accurate results than PH-ICNN under both input scenarios, indicating the benefit of relaxing the convexity constraint in our method.

\subsection{Double pendulum system}

To validate the capability of the proposed method in preserving multiple stable equilibria,
we consider an autonomous double pendulum system with highly nonlinear dynamics and a non-convex Hamiltonian, leading to infinitely many equilibria.
In addition, comparisons with PH-ICNN are conducted to demonstrate superior performance of the proposed method in modeling the system.

\begin{figure}[!tb]
    \centering
    \includegraphics[width=0.55\linewidth]{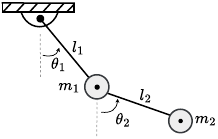}
    \caption{A double pendulum system.}
    \label{fig:DoubPen}
\end{figure}

\begin{figure*}[!t]
\subfloat[$\mathbf{x}_0^{(1)}$, $\xeq^{(1)}$\label{fig:3DoubPen:a}]{\includegraphics[width=0.33 \linewidth]{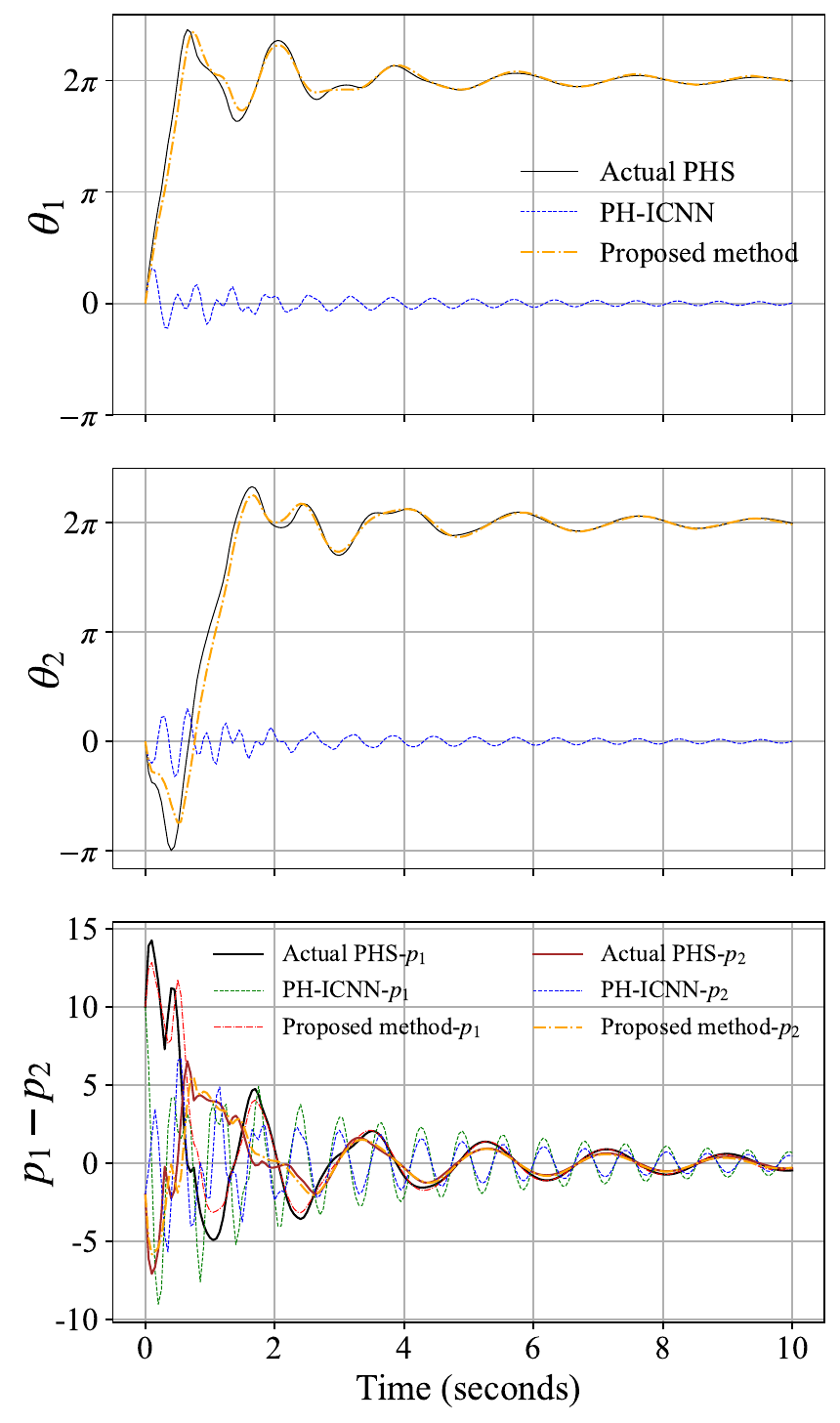}
}
\subfloat[$\mathbf{x}_0^{(2)}$, $\xeq^{(2)}$\label{fig:3DoubPen:b}]{\includegraphics[width=0.33 \linewidth]{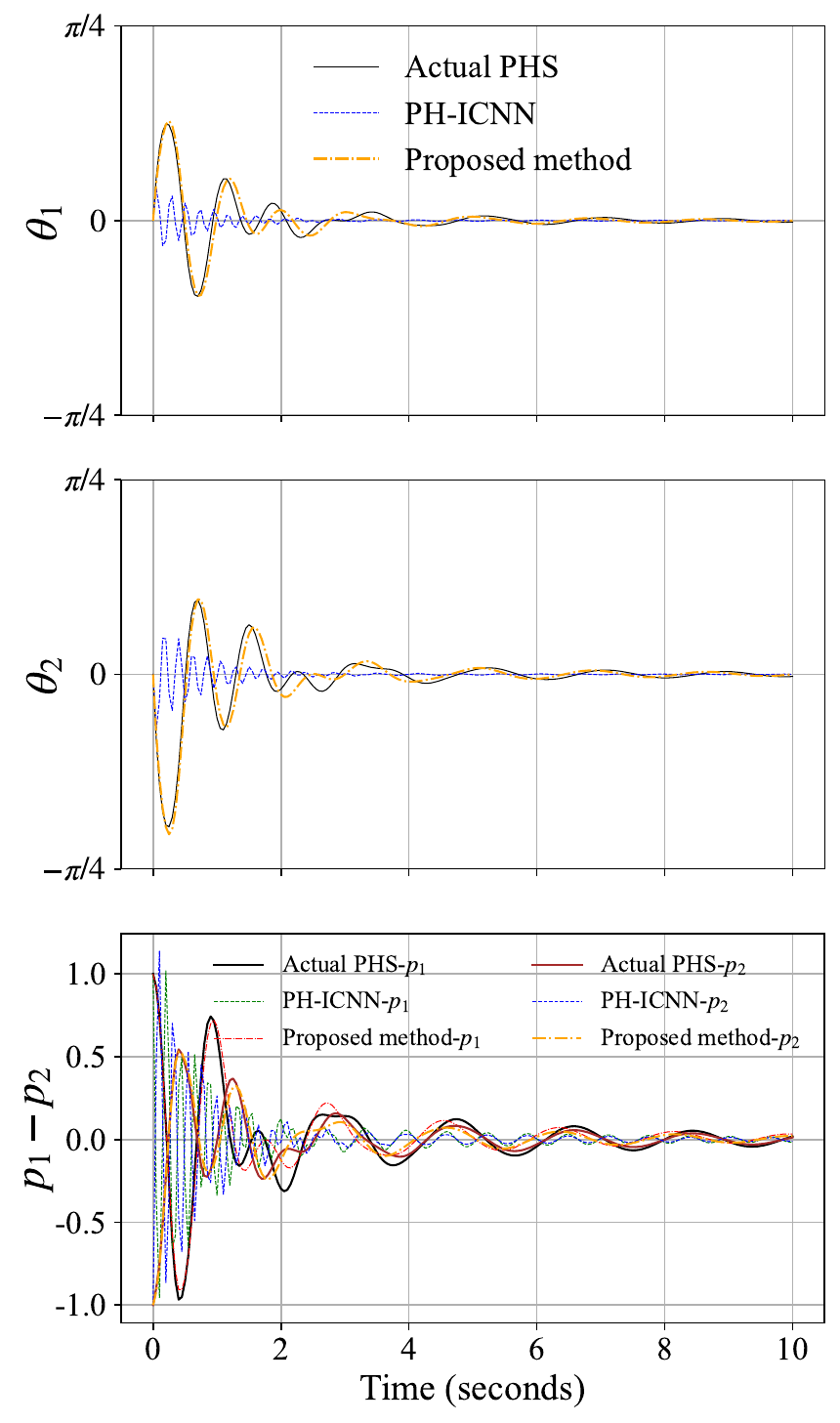}}
\subfloat[$\mathbf{x}_0^{(3)}$, $\xeq^{(3)}$\label{fig:3DoubPen:c}]{\includegraphics[width=0.33 \linewidth]{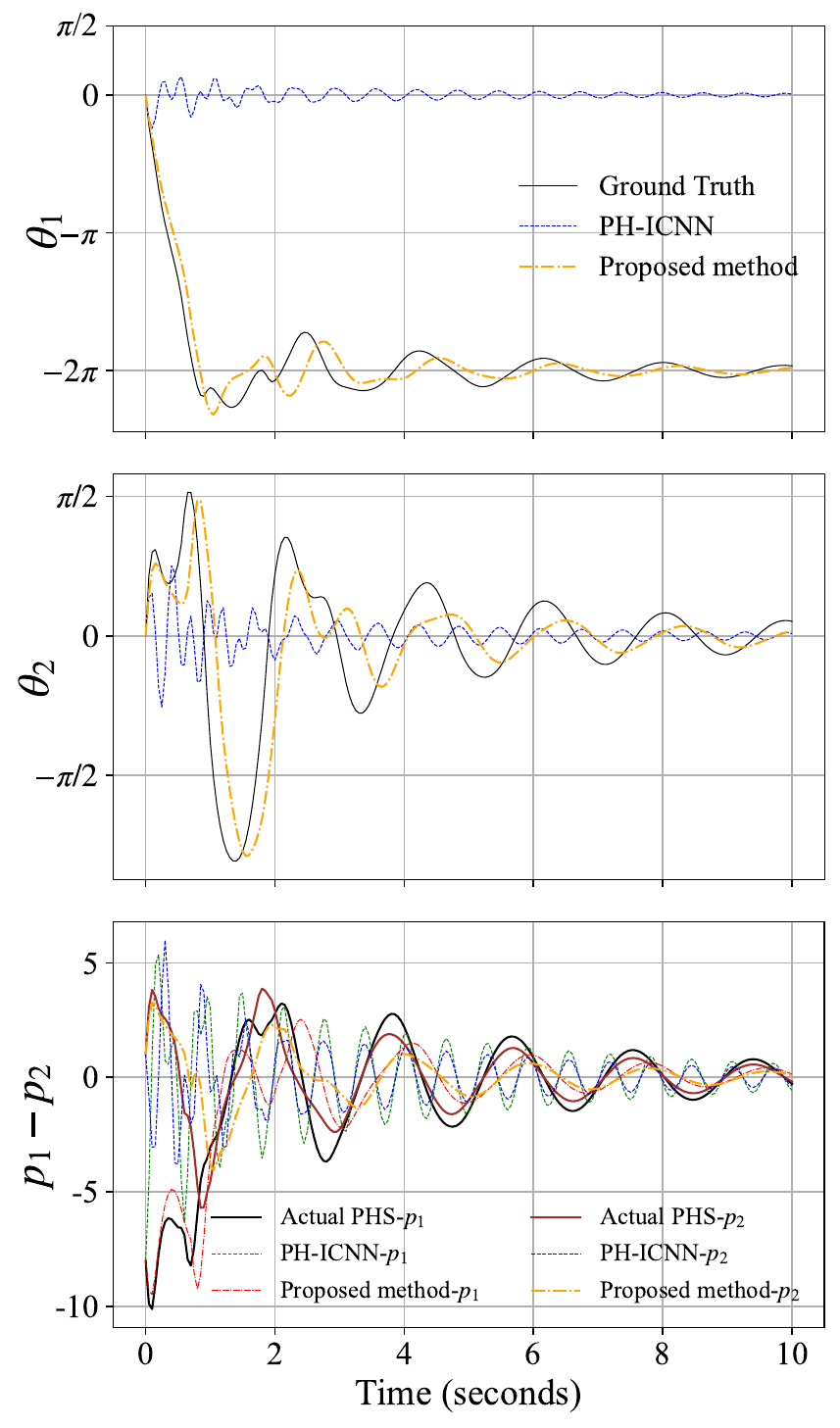}}
\caption{State trajectories of the double pendulum system from three different initial states 
$\mathbf{x}_0^{(1)} = [0, 0, 10, -2]^\top$, 
$\mathbf{x}_0^{(2)} = [0, 0,  1, -1]^\top$, 
and
$\mathbf{x}_0^{(3)} = [0, 0, -8,  1]^\top$.
These cases illustrate different stable equilibria $\xeq^{(1)} = [2\pi, 2\pi, 0, 0]^\top$,
$\xeq^{(2)} = [0, 0, 0, 0]^\top$, and
$\xeq^{(3)} = [-2\pi, 0, 0, 0]^\top$.}
\label{fig:3DoubPen}
\end{figure*}

Let $m_1$ and $m_2$ denote the masses of the two links,
$l_1$ and $l_2$ their corresponding lengths, and
$\theta_1$ and $\theta_2$ the angular positions of the links, as illustrated in Fig.~\ref{fig:DoubPen}.
Here, we consider the autonomous case where the control input is $\mathbf{u}=\bm{0}$.
The matrices $J$ and $R$ are given by
\begin{align*}
J &= \begin{bmatrix}
\bm{0} & I_2 \\
-I_2 & \bm{0}
\end{bmatrix}, \qquad
R = \begin{bmatrix}
\bm{0} & \bm{0} \\
\bm{0} & \mathrm{diag}(\gamma_1, \gamma_2)
\end{bmatrix},
\end{align*}
where $I_2$ denotes the $2 \times 2$ identity matrix, and $\gamma_1$ and $\gamma_2$ are the known damping coefficients of the double pendulum system.
The state vector is $\mathbf{x} = [\mathbf{q}^\top,\; \mathbf{p}^\top]^\top$, where
$\mathbf{q} = [\theta_1,\;\theta_2]^\top$ represents the generalized coordinates and
$\mathbf{p} = [p_1,\;p_2]^\top$ the corresponding generalized momenta. 
In this system, the Hamiltonian $H(\mathbf{x})$ represents the total energy of the system and is presented by
\begin{align*}
H(\mathbf{x}) =
\frac{1}{2}\mathbf{p}^\top M^{-1}(\mathbf{q}) \mathbf{p} + V(\mathbf{q}),
\end{align*}
where $\frac{1}{2}\mathbf{p}^\top M^{-1}(\mathbf{q}) \mathbf{p}$ corresponds to the kinetic energy and $V(\mathbf{q})$ denotes the potential energy.
The mass matrix $M(\mathbf{q})$ and the potential energy $V(\mathbf{q})$ are defined as follows
\begin{align*}
M(\mathbf{q}) &=
\begin{bmatrix}
(m_1+m_2)l_1^2 & m_2 l_1 l_2 \cos(\theta_1-\theta_2) \\
m_2 l_1 l_2 \cos(\theta_1-\theta_2) & m_2 l_2^2
\end{bmatrix}\text,
\\
V(\mathbf{q}) &=
-(m_1+m_2)g l_1 \cos\theta_1
-
m_2 g l_2 \cos\theta_2 .
\end{align*}
The initial state is
$\mathbf{x}_0 = [\mathbf{q}_0^\top, \mathbf{p}_0^\top]^\top$, and 
the stable equilibria $\mathbf{x}_\mathrm{eq} = [\mathbf{q}_\mathrm{eq}^\top, \; \mathbf{p}_\mathrm{eq}^\top]^\top$ are of the form
\begin{align*}
    \mathbf{q}_\mathrm{eq} = \begin{bmatrix}
        2n_z \pi \\
        2m_z\pi
    \end{bmatrix}\text, \quad
    \mathbf{p}_\mathrm{eq} = \begin{bmatrix}
        0 \\ 
        0
    \end{bmatrix} \text, \quad n_z, m_z \in \mathbb{Z} \text.
\end{align*}
Unstable equilibria are not considered in this example.

Training data are generated by simulating the double pendulum system from multiple initial conditions. 
The initial angles are fixed to $\mathbf{q}_0 = [0,\,0]^\top$, while the initial momenta are sampled over the domain $\mathbf{p}_0 \in [-10,\,10] \times [-2,\,2]$. 
A $10 \times 10$ mesh grid is used to generate the set of initial states.
The damping coefficients are set to $\gamma_1 = \gamma_2 = 0.5$ and are assumed to be known.
The physical parameters of the double pendulum are $m_1 = m_2 = \qty{2}{\kilogram}$ and $l_1 = l_2 = \qty{0.5}{\meter}$.
For each initial condition, the system is simulated over a time horizon of $\qty{10}{\second}$ with a sampling interval $\Delta t = \qty{0.05}{\second}$, yielding 200 samples per trajectory.
In total, the training dataset contains \num{20000} samples.
The NNs used to learn the Hamiltonian $H(\mathbf{x})$ has two hidden layers, each with 64 nodes. 
The function $\sigma(\cdot)$ is structured with the order $d = 2$ and the radius $b = 0.5$.
Nine equilibrium points are enforced in our proposed model, given by $\mathbf{x}_\mathrm{eq} = \{ \tilde{\theta}_1, \tilde{\theta}_2, 0, 0 \}$, where $\tilde{\theta}_1, \tilde{\theta}_2 \in \{-2\pi, 0, 2\pi\}$.

Fig.~\ref{fig:3DoubPen} show the trajectories of the angles $\theta_1$ and $\theta_2$ for three cases with different initial states of the double pendulum system.
For different initial states, the system converges to different equilibrium points.

The first case in Fig.~\ref{fig:3DoubPen:a} corresponds to a configuration in which both links stabilize at non–origin equilibria.
In this scenario, the PH-ICNN model fails to reproduce the correct trajectories and instead converges to the origin.
This behavior arises from the structural property of the PH-ICNN model, which enforces a single equilibrium point.
In contrast, the proposed model successfully converges to the correct equilibrium $[\theta_1, \theta_2, p_1, p_2] = [2\pi, 2\pi, 0, 0]$ and accurately reproduces the trajectories compared with the ground truth.

For the second case in Fig.~\ref{fig:3DoubPen:b}, where the equilibrium is located at the origin, the proposed model still provides a more accurate approximation of the system dynamics than does the PH-ICNN model.
Although the PH-ICNN model eventually converges to the correct equilibrium, it produces incorrect trajectories during the initial transient period (approximately the first 5 seconds).

Finally, in the third case shown in Fig.~\ref{fig:3DoubPen:c}, where the equilibrium is at $\theta_{\mathrm{eq}}^{(1)} = -2\pi$ and $\theta_{\mathrm{eq}}^{(2)} = 0$, the proposed model again closely matches the ground truth dynamics, whereas PH-ICNN fails to capture the correct system behavior.

\section{Conclusion}
\label{sec:conclusion}
This paper addresses the problem of modeling port-Hamiltonian (PH) systems from data while preserving the PH structure and ensuring system stability by incorporating stable equilibria information into the learning process.
Theoretical analysis shows that the proposed data-driven approach guarantees stability at the specified equilibria while avoiding the convexity constraints imposed by ICNN-based Hamiltonian learning.
The effectiveness of the proposed method is demonstrated by two PH systems with distinct characteristics: one with a single stable equilibrium and the other with multiple stable equilibria.
The results show that the proposed model captures more accurately the system dynamics and significantly improves performance compared with the PH-ICNN approach.
Although the proposed method achieves promising results in the considered benchmarks, its applicability to real-world systems has not yet been investigated.
Future work will extend the proposed method to control applications and evaluate its performance on practical physical systems.

\section*{Appendix}
\subsection{Proof of Lemma~\ref{lem:unique} \label{sec:proofLem1}}
Because $H$ is continuously differentiable and $\xeq$ is a strict local minimum of $H$, 
the first-order necessary condition gives
$\gradH(\xeq)=\mathbf{0}$. 
In addition, consider $\Bc(\xeq,\varepsilon)$ as in \eqref{eq:strictly} since $\xeq$ is a strict local minimum of $H$,
let $\xbf \in \Bc(\xeq,\varepsilon)$ satisfy
$\gradH(\xbf)= \mathbf{0}$.
We want to prove that $\xbf = \xeq$.
Assume by contradiction that $\xbf \neq \xeq$ and $\gradH(\xbf)=\mathbf{0}$.

Define
$G(\xbf):=H(\xbf)- \alpha(\|\xbf-\xeq\|_2)$.
Then for all $\xbf\in B(\xeq,\varepsilon)\setminus \{\xeq\}$,
$G(\xbf) > G(\xeq)$,
because
$H(\xbf) \ge H(\xeq)+ \alpha(\|\xbf-\xeq\|_2)$.
Thus $\xeq$ is a local minimum of $G$ on the ball. 
Since $G$ is continuously differentiable, we have
$\nabla_\xbf \mathbf G(\xeq)= \mathbf{0}$.  
Consider the line segment $\boldsymbol{\ell}(t)=\xeq+t(\xbf-\xeq), t\in[0,1]$.
Because the ball is convex, $\boldsymbol{\ell}(t) \in B(\xeq,\varepsilon)$.
Define $\phi(t):=H(\boldsymbol{\ell}(t))$, then $\phi(t)$ is continuously differentiable, and by strict local minimality \eqref{eq:strictly},
$\phi(t)\ge \phi(0) +  \alpha(t\|\xbf-\xeq\|_2), \forall t\in[0,1]$.
Moreover, $\frac{d \phi}{d t}(0)=\gradH(\xeq)^\top (\xbf-\xeq)=0$ and
$\frac{d \phi}{d t}(1)=\gradH(\xbf)^\top (\xbf-\xeq)=0$ since $\gradH(\xbf) = \gradH(\xeq)= \mathbf{0}$.

Now apply Rolle’s theorem to $\frac{d \phi}{d t}$ informally via the growth condition: since
$\phi(t)-\phi(0) \ge \alpha(t\|\xbf-\xeq\|_2$),
the function $\phi(t)$ must increase strictly away from $0$ as $t$ increases, and therefore it cannot return to a critical point at $t=1$ unless $\xbf = \xeq$. Otherwise, $t=1$ would correspond to another stationary point inside the same neighborhood, contradicting the growth from the $\alpha$.
Hence $\xbf=\xeq$, and $\xeq$ is the unique solution of
$\gradH(\xbf)= \mathbf{0}$
in $\Bc(\xeq,\varepsilon)$.

\subsection{Proof of Lemma \ref{lem:local_stable} \label{sec:proofLem2}}
  Intuitively, $\xeq$ is an equilibrium point of \eqref{eq:aPHS} since $\dot\xbf = {\bf 0}$ at $\xbf = \xeq$ due to {Assumption~\ref{ass_hamil}}.
  We choose $H(\xbf)$ as the Lyapunov function as in \eqref{eq:energy}, which is non-increasing over time.
  Since $\xeq$ is a strict minimum of $H(\xbf)$, let us define 
  $\alpha$
  and $\Bc(\xeq,\varepsilon)$ as in \eqref{eq:strictly}.
  We will prove two statements:
  \begin{itemize}
    \item[(i)]  $\xbf(t) = \xeq$ is a unique solution of $\gradH(\xbf) = \mathbf{0}$ over the $\Bc(\xeq,\varepsilon)$ in \eqref{eq:strictly}; 
    \item[(ii)] For any initial state $\xbf(0) \neq \xeq$ in a subset of $ \Bc(\xeq,\varepsilon)$, $\xbf(t)$ remains in $\Bc(\xeq,\varepsilon)$.
  \end{itemize}
  First, (i) is satisfied due to {Assumption~\ref{ass_hamil}} and Lemma \ref{lem:unique}.
  For (ii), consider the level set $\Sc_{H\delta} = \{\xbf \in \Bc(\xeq,\varepsilon)| H(\xbf) \leq \delta_\varepsilon + H(\xeq)\}$, where $\delta_\varepsilon$ is a constant depending on $\varepsilon$.
  From \eqref{eq:strictly}, we have $\Sc_{H\delta} \subseteq \Sc_{\alpha\delta} = \{\xbf \in \Bc(\xeq,\varepsilon)| \alpha(\norm{\xbf - \xeq}_2) \leq \delta_\varepsilon\} \subseteq \Bc(\xeq,\varepsilon)$, which are all bounded.
  Consequently, every $\xbf(t)$ starting in $\Sc_{H\delta}$ will remain in $\Sc_{H\delta} \subseteq \Bc(\xeq,\varepsilon)$ since $\dot H(\xbf)$ is non-positive.
  Therefore, $\xeq$ is a stable equilibrium according to {Definition \ref{def}}.
  According to LaSalle's invariant set theorem \cite{krstic1995nonlinear}, all solutions of \eqref{eq:PHS} will converge to $\Om = \{\xbf| \big(\gradH(\xbf)\big)^\top R(\xbf) \big(\gradH(\xbf)\big) = \mathbf{0}\} $.
  Moreover, from {Lemma \ref{lem:unique}}, if $R(\xeq) \succ 0$, then there exists $\varepsilon > 0$ such that $\big(\gradH(\xbf)\big)^\top R(\xbf) \big(\gradH(\xbf)\big) \succ 0$ for all $\xbf \in B(\xeq,\varepsilon) \setminus \xeq$.
  Hence, if $\xbf(t)$ starts in $S_{H\delta}$, it remains in $S_{H\delta}$ and converges to $\xeq$. 
  Therefore, $\xeq$ is an asymptotically stable equilibrium.

\bibliographystyle{IEEEtran}
\bibliography{RefPHs}

@inproceedings{altawaitanHamiltonianDynamicsLearning2024,
  title = {Hamiltonian Dynamics Learning from Point Cloud Observations for Nonholonomic Mobile Robot Control},
  booktitle = {2024 {{IEEE Int. Conf.}} on {{Robotics}} and {{Automation}} ({{ICRA}})},
  author = {Altawaitan, Abdullah and Stanley, Jason and Ghosal, Sambaran and Duong, Thai and Atanasov, Nikolay},
  year = {2024},
  month = may,
  pages = {16937--16944},
  doi = {10.1109/ICRA57147.2024.10610395},
  urldate = {2024-12-20}
}

@inproceedings{o2025port,
  title={A port-{Hamiltonian} formulation of mechanical systems with switching contact constraints},
  author={O’Brien, Thomas and Ferguson, Joel and Donaire, Alejandro},
  booktitle={2025 European Control Conf. (ECC)},
  pages={1918--1924},
  year={2025},
  organization={IEEE}
}

@inproceedings{bartel2022port,
  title={Port-{Hamiltonian} systems’ modelling in electrical engineering},
  author={Bartel, Andreas and Clemens, Markus and G{\"u}nther, Michael and Jacob, Birgit and Reis, Timo},
  booktitle={Int. Conf. on Scientific Computing in Electrical Engineering},
  pages={133--143},
  year={2022},
  organization={Springer}
}

@inproceedings{lee2012dynamics,
  title={Dynamics and control of a chain pendulum on a cart},
  author={Lee, Taeyoung and Leok, Melvin and McClamroch, N Harris},
  booktitle={2012 IEEE 51st IEEE Conf. on Decision and Control (CDC)},
  pages={2502--2508},
  year={2012},
  organization={IEEE}
}

@misc{sanchezescalonilla2024robustneuralidapbcpassivitybased,
      title={Robust Neural {IDA-PBC}: passivity-based stabilization under approximations}, 
      author={Santiago Sanchez-Escalonilla and Samuele Zoboli and Bayu Jayawardhana},
      year={2024},
      eprint={2409.16008},
      archivePrefix={arXiv},
      primaryClass={eess.SY},
      url={https://arxiv.org/abs/2409.16008}, 
}

@article{ge2004position,
  title={Position control of chained multiple mass-spring-damper systems-Adaptive output feedback control approaches},
  author={Ge, SS and Huang, L and Lee, TH},
  journal={Int. Journal of Control Automation and Systems},
  volume={2},
  pages={144--155},
  year={2004},
  publisher={Korean Institute of Electrical Engineers}
}

@article{rettberg2025data,
  title={Data-driven identification of latent port-{Hamiltonian} systems},
  author={Rettberg, Johannes and Kneifl, Jonas and Herb, Julius and Buchfink, Patrick and Fehr, J{\"o}rg and Haasdonk, Bernard},
  journal={Computational Science and Engineering},
  volume={2},
  number={1},
  pages={4},
  year={2025},
  publisher={Springer}
}

@article{cherifi2025nonlinear,
  title={Nonlinear port-{Hamiltonian} system identification from input-state-output data},
  author={Cherifi, Karim and Messaoudi, Achraf El and Gernandt, Hannes and Roschkowski, Marco},
  journal={arXiv preprint arXiv:2501.06118},
  year={2025}
}

@article{nageshrao2015port,
  title={Port-{Hamiltonian} systems in adaptive and learning control: {A} survey},
  author={Nageshrao, Subramanya P and Lopes, Gabriel AD and Jeltsema, Dimitri and Babu{\v{s}}ka, Robert},
  journal={IEEE Trans. on Automatic Control},
  volume={61},
  number={5},
  pages={1223--1238},
  year={2015},
  publisher={IEEE}
}

@inproceedings{amos2017input,
  title={Input convex neural networks},
  author={Amos, Brandon and Xu, Lei and Kolter, J Zico},
  booktitle={Int. Conf. on machine learning},
  pages={146--155},
  year={2017},
  organization={PMLR}
}

@article{kotyczka2019discrete,
  title={Discrete-time port-{Hamiltonian} systems: {A} definition based on symplectic integration},
  author={Kotyczka, Paul and Lefevre, Laurent},
  journal={Systems \& Control Letters},
  volume={133},
  pages={104530},
  year={2019},
  publisher={Elsevier}
}

@inproceedings{rothStablePortHamiltonianNeural2025,
  title={Stable Port-{Hamiltonian} Neural Networks},
  author={Roth, Fabian J and Klein, Dominik K and Kannapinn, Maximilian and Peters, Jan and Weeger, Oliver},
  booktitle={The Thirty-ninth Annual Conf. on Neural Information Processing Systems},
  year={2025}
}

@article{gengDataDrivenReducedOrderModels2025,
  title = {Data-driven reduced-order models for port-{Hamiltonian} systems with operator inference},
    journal = {Computer Methods in Applied Mechanics and Engineering},
    volume = {442},
    pages = {118042},
    year = {2025},
    issn = {0045-7825},
}

@article{schwerdtnerAdaptiveSamplingStructurePreserving2021,
  title = {Adaptive {{Sampling}} for {{Structure-Preserving Model Order Reduction}} of {{Port-Hamiltonian Systems}}⁎},
  author = {Schwerdtner, Paul and Voigt, Matthias},
  year = {2021},
  month = jan,
  journal = {IFAC-PapersOnLine},
  series = {7th {{IFAC Workshop}} on {{Lagrangian}} and {{Hamiltonian Methods}} for {{Nonlinear Control LHMNC}} 2021},
  volume = {54},
  number = {19},
  pages = {143--148},
  issn = {2405-8963},
  doi = {10.1016/j.ifacol.2021.11.069},
  urldate = {2025-03-30}
}

@incollection{vanderschaftPortHamiltonianSystemsNetwork2004,
  title = {Port-{{Hamiltonian Systems}}: {{Network Modeling}} and {{Control}} of {{Nonlinear Physical Systems}}},
  shorttitle = {Port-{{Hamiltonian Systems}}},
  booktitle = {Advanced {{Dynamics}} and {{Control}} of {{Structures}} and {{Machines}}},
  author = {{van der Schaft}, A. J.},
  editor = {Irschik, Hans and Schlacher, Kurt},
  year = {2004},
  pages = {127--167},
  publisher = {Springer},
  address = {Vienna},
  doi = {10.1007/978-3-7091-2774-2_9},
  urldate = {2025-03-31},
  isbn = {978-3-7091-2774-2},
  langid = {english}
}

@book{khalil2002nonlinear,
  title={Nonlinear systems},
  author={Khalil, HK},
  publisher={Prentice Hall},
  year={2002}
}

@article{nguyen2020distributed,
  title={Distributed flocking bounded control of second-order dynamic multiple polygonal agents},
  author={Nguyen, Thanh Binh and Kim, Sung Hyun},
  journal={IEEE Access},
  volume={8},
  pages={200170--200179},
  year={2020},
  publisher={IEEE}
}

@book{krstic1995nonlinear,
  title={Nonlinear and Adaptive Control Design},
  author={Krstic, M},
  publisher={John Willey, New York},
  year={1995}
}

@article{desaiPortHamiltonianNeuralNetworks2021,
  title = {Port-{{Hamiltonian}} Neural Networks for Learning Explicit Time-Dependent Dynamical Systems},
  author = {Desai, Shaan A. and Mattheakis, Marios and Sondak, David and Protopapas, Pavlos and Roberts, Stephen J.},
  year = {2021},
  month = sep,
  journal = {Physical Review E},
  volume = {104},
  number = {3},
  pages = {034312},
  issn = {2470-0045, 2470-0053},
  doi = {10.1103/PhysRevE.104.034312},
  urldate = {2024-12-06},
  langid = {english}
}

@article{duongPortHamiltonianNeuralODE2024,
  title = {Port-{{Hamiltonian}} Neural {{ODE}} Networks on Lie Groups for Robot Dynamics Learning and Control},
  author = {Duong, Thai and Altawaitan, Abdullah and Stanley, Jason and Atanasov, Nikolay},
  year = {2024},
  journal = {IEEE Trans. on Robotics},
  volume = {40},
  pages = {3695--3715},
  issn = {1941-0468},
  doi = {10.1109/TRO.2024.3428433},
  urldate = {2024-12-05}
}

@article{eidnesPseudoHamiltonianNeuralNetworks2023,
  title = {Pseudo-{Hamiltonian} neural networks with state-dependent external forces},
journal = {Physica D: Nonlinear Phenomena},
volume = {446},
pages = {133673},
year = {2023},
issn = {0167-2789},
author = {Sølve Eidnes and Alexander J. Stasik and Camilla Sterud and Eivind Bøhn and Signe Riemer-Sørensen}
}

@article{rashadEnergyAwareImpedance2022,
  title = {Energy Aware Impedance Control of a Flying End-Effector in the Port-{{Hamiltonian}} Framework},
  author = {Rashad, Ramy and Bicego, Davide and Zult, Jelle and {Sanchez-Escalonilla}, Santiago and Jiao, Ran and Franchi, Antonio and Stramigioli, Stefano},
  year = {2022},
  month = dec,
  journal = {IEEE Trans. on Robotics},
  volume = {38},
  number = {6},
  pages = {3936--3955},
  issn = {1941-0468},
  doi = {10.1109/TRO.2022.3183532},
  urldate = {2024-12-13}
}

@article{rashadTwentyYearsDistributed2020,
  title = {Twenty Years of Distributed Port-{{Hamiltonian}} Systems: A Literature Review},
  shorttitle = {Twenty Years of Distributed Port-{{Hamiltonian}} Systems},
  author = {Rashad, Ramy and Califano, Federico and {van der Schaft}, Arjan J and Stramigioli, Stefano},
  year = {2020},
  month = dec,
  journal = {IMA Journal of Mathematical Control and Information},
  volume = {37},
  number = {4},
  pages = {1400--1422},
  issn = {1471-6887},
  doi = {10.1093/imamci/dnaa018},
  urldate = {2024-12-10}
}

@article{schaftPortHamiltonianSystemsTheory2014,
  title = {Port-{{Hamiltonian Systems Theory}}: {{An Introductory Overview}}},
  shorttitle = {Port-{{Hamiltonian Systems Theory}}},
  author = {van der Schaft, Arjan and Jeltsema, Dimitri},
  year = {2014},
  month = jun,
  journal = {Foundations and Trends{\textregistered} in Systems and Control},
  volume = {1},
  number = {2-3},
  pages = {173--37 8},
  publisher = {Now Publishers, Inc.},
  issn = {2325-6818, 2325-6826},
  doi = {10.1561/2600000002},
  urldate = {2024-12-11},
  langid = {english}
}

@article{zhongPortHamiltonianControlFramework2022,
  title = {A {{Port-Hamiltonian Control Framework}} to {{Render}} a {{Power Electronic System Passive}}},
  author = {Zhong, Qing-Chang and Stefanello, M{\'a}rcio},
  year = {2022},
  month = apr,
  journal = {IEEE Trans. on Automatic Control},
  volume = {67},
  number = {4},
  pages = {1960--1965},
  issn = {1558-2523},
  doi = {10.1109/TAC.2021.3069389},
  urldate = {2025-01-06}
}

\end{document}